\documentclass{article}

\def\notes{0} 

\usepackage[utf8]{inputenc}
\usepackage{fullpage}
\usepackage{amsfonts,amssymb,amsmath,amsthm}

\newcommand{\E}{\mathbb{E}}
\newcommand{\bern}{\text{Bern}}

\usepackage{amsmath,amssymb}
\newif\ifevanfancy\evanfancytrue
\newif\ifevanhdr\evanhdrtrue
\newif\ifevanhref\evanhreftrue
\newif\ifevansetup\evansetuptrue
\newif\ifevanthm\evanthmtrue
\newif\ifevansecthm\evansecthmfalse
\newif\ifevanht\evanhtfalse
\newif\ifevanpkg\evanpkgtrue
\newif\ifevanpdf\evanpdftrue
\newif\ifevantitling\evantitlingtrue
\newif\ifevanauthor\evanauthortrue
\newif\ifevanchinese\evanchinesefalse
\newif\ifevanmdthm\evanmdthmfalse
\newif\ifevandiagrams\evandiagramsfalse
\newif\ifevanpatchasy\evanpatchasyfalse
\newif\ifevanhints\evanhintsfalse
\newif\ifevanasy\evanasytrue
\newif\ifevancolorsec\evancolorsecfalse
\newif\ifevantitlemark\evantitlemarktrue

\relax

\ifevanpkg\else\evanmdthmfalse\fi
\ifevanpkg\else\evanasyfalse\fi

\ifevansetup\else\evanthmfalse\fi

\newcommand{\eps}{\varepsilon}

\newcommand{\BP}{\mathbb P}

	{\end{proof}}

	{\end{mdframed}}

\ifevanasy
\ifevanpatchasy
\usepackage{patch-asy}
\else
\usepackage{asymptote}
\fi
\begin{asydef}
	import olympiad;
	import cse5;
	pointpen = black;
	pathpen = black;
	pathfontpen = black;
	anglepen = black;
	anglefontpen = black;
	pointfontsize = 10;
	defaultpen(fontsize(10pt));
	size(8cm); // set a reasonable default
	usepackage("amsmath");
	usepackage("amssymb");
	settings.tex="latex";
	settings.outformat="pdf";
\end{asydef}
\fi
\ifevanthm
\usepackage{amsthm}
\fi

\newtheorem{theorem}{Theorem}[section]
\newtheorem{definition}[theorem]{Definition}
\newtheorem{proposition}[theorem]{Proposition}
\newtheorem{lemma}[theorem]{Lemma}
\newtheorem{claim}[theorem]{Claim}

\usepackage{xcolor}

\ifnum\notes=1
\newcommand{\mnote}[1]{{\bf \color{red} Madhu: #1}}
\newcommand{\dnote}[1]{{\bf \color{red} David: #1}}
\newcommand{\snote}[1]{{\color{purple} Small Note: #1}}
\else
\newcommand{\mnote}[1]{}
\newcommand{\dnote}[1]{}
\newcommand{\snote}[1]{}
\fi

\newcommand{\calm}{\mathcal{M}}
\newcommand{\Huff}{\mathrm{Huff}}

\newcommand{\IH}{\mathrm{IH}}
\newcommand{\LZ}{\mathrm{LZ}}

\title{A Self-contained Analysis of the Lempel-Ziv Compression Algorithm}

\author{Madhu Sudan\thanks{Harvard John A. Paulson School of Engineering and Applied Sciences, Harvard University, 33 Oxford Street,
Cambridge, MA 02138, USA. {\tt madhu@cs.harvard.edu}. Work supported in part by a Simons Investigator Award and
NSF Award CCF 1715187.} \and David Xiang\thanks{Harvard College. {\tt davidxiang@college.harvard.edu}}}

\date{October 2, 2019}

\begin{document}

\maketitle

\begin{abstract}
This article gives a self-contained analysis of the performance of the Lempel-Ziv compression algorithm on (hidden) Markovian sources. Specifically we include a full proof of the assertion that the compression rate approaches the entropy rate of the chain being compressed.
\end{abstract}

\section{Introduction}

In the late 1970's Abraham Lempel and Jacob Ziv~\cite{LZ1,LZ2,LZ3} gave some extremely simple, clever and efficient algorithms that were able to {\em universally} compress outputs of ``nice'' stochastic processes down to their entropy rate. While their algorithms are well-known and understood, the analysis of their algorithms is not widely understood. The aim of this article is to remedy this situation by providing a self-contained statement and analysis of their algorithm for the special case of ``(hidden) Markov models''. Our primary hope is that this article can form the basis of lectures in undergraduate courses that teach this algorithm along with analysis to students across a broad spectrum of disciplines. In particular our analysis depends only on elementary discrete probability theory (as used in say~\cite{MitzUpf,MotRagh}), basic facts about Markov chains (e.g., \cite[Chapter 1]{LevinPW}) and elementary information theory~\cite[Chapter 2]{CoverThomas}. 
The proofs here are essentially the same as those in the original articles though the actual exposition is from scratch and the specific analysis we use here goes back to unpublished notes of Bob Gallager from the 1990s~\cite{RG90}. In particular we owe our understanding of the overview of the entire analysis, as well as most of the specific notions and claims of Section~\ref{sec:LZ-anal}, to these notes.

We now turn to stating the main theorem we wish to prove. In Section~\ref{ssec:def} we introduce some the basic terminology that will allow us to state the main result. Specifically we recall the notion of a finite state Markov chain, define a hidden Markov model, and define its entropy rate. We also define what it means for a compression algorithm to be universal (for the class of hidden Markov models). In Section~\ref{ssec:lz-alg} we then give a version of the Lempel-Ziv algorithm. And in Section~\ref{ssec:main-thm} we state the main theorem about this algorithm, namely that the version of the Lempel-Ziv algorithm we describe is universal for hidden Markov models.

\subsection{Definitions}\label{ssec:def}

A random sequence $Z_0,Z_1,Z_2,\ldots$ with $Z_t \in [k]$ is said to be a (time-invariant) finite state Markov chain if there is a $k \times k$ matrix $M$ such that for every $t\geq 0$ and $i,j \in [k]$ it is the case that $\Pr[Z_{t+1} = j | Z_t = i] = M_{ij}$. Such a Markov chain $M$ is said to be specified by the matrix $M$ and the distribution $\Pi_0 \in \Delta([k])$ of the random variable $Z_0$. We say that $M$ is a $k$-state Markov chain.

The Markov chain is thus essentially specified by a weighted directed graph (possibly with self loops) corresponding to the matrix $M$. The chain is said to be irreducible if the underlying graph is strongly connected. Equivalently a chain $Z_0,\ldots,Z_t,\ldots$ is irreducible if for every $i,j\in[k]$ there exists $t$ such that $\Pr[Z_t = j | Z_0 = i] > 0$. Similarly a Markov chain is said to be aperiodic if greatest common divisor of the cycle lengths in the underlying graphs is $1$. Formally we say that the chain $Z_0,Z_1,Z_2,\ldots$
has a cycle of length $\ell>0$ if there exists $t$ and state $i$ such that $\Pr[Z_{t+\ell} = i | Z_t = i] > 0$. The chain is said to be aperiodic if the greatest common divisor of cycle lengths is $1$. 

Throughout this article we will consider only irreducible aperiodic Markov chains. (The study can be extended to the periodic case easily, but the irreducible case is actually different.) Irreducible and aperiodic chains have a unique stationary distribution $\Pi$ (satisfying $\Pi = \Pi \cdot M$).

\begin{definition}[Hidden Markov Models] A sequence $X_0,X_1,X_2,\ldots,$ with $X_t \in \Sigma$ is said to be a {\em Hidden Markov Model} if there exists a $k$-state Markov chain $Z_0,Z_1,Z_2,\ldots$, specified by matrix $M$ and initial distribution $\Pi_0$, and $k$ distributions $P^{(1)},\ldots,P^{(k)} \in \Delta(\Sigma)$ such that $X_t \sim P^{(Z_t)}$. 
We use the notation $\calm$ to denote the ingredients specifying a hidden Markov model, namely the tuple $(k,M,\Pi_0,P^{(1)},\ldots,P^{(k)})$. 
\end{definition}

\begin{definition}[Entropy Rate] The entropy rate of a hidden Markov model $\calm$ with output denoted $X_0,X_1,\ldots,$ is given by $\lim_{t \to \infty} \left\{\frac1 t \cdot H(X_0,\ldots,X_{t-1})\right\}$ when the limit exists. The entropy rate is denoted $H(\calm)$.
\end{definition}

\begin{proposition}
The entropy rate (i.e., the limit) always exists for irreducible, aperiodic, Markov chains. Furthermore it does not depend on the initial distribution $\Pi_0$.
\end{proposition}

\begin{proof}
We start with the ``furthermore'' part. This part follows from the fact that irreducible aperiodic chains converge to their stationary distribution, i.e., for every $\eps > 0$  there exists $T = T(\eps)$ such that for all $Z_0$ and every $t \geq T_0$ the distribution of $X_t$ is $\epsilon$-close (in statistical distance) to the stationary probability distribution. Now consider $n$ sufficiently large.
We have $$H(X_t,\ldots,X_n) \leq H(X_0,\ldots,X_n) \leq H(X_t,\ldots,X_n) + t\log |\Omega|.$$ Now, we'll compare $H(X_t, \dotsc X_n) = H(X_t) + \sum_{k=t}^{n-1} H(X_{k+1}|X_k)$ to $H(X'_t, \dotsc X'_n) = H(X'_t) + \sum_{k=t}^{n-1} H(X'_{k+1}|X'_k)$, where $X'_t, \dotsc X'_n$ are obtained from starting with $X'_t$ as the stationary probability distribution. We first bound the distance $|H(X'_{k+1}|X'_k)-H(X_{k+1}|X_k)|$ by noting 
\begin{align*}
    |H(X'_{k+1}|X'_k)-H(X_{k+1}|X_k)| &= |\sum_{i \in \Sigma} \Pr(X'_k=i)H(X'_{k+1}|X'_k=i)-\Pr(X_k=i)H(X_{k+1}|X_k=i)|\\
    &\leq \sum_{i\in \Sigma} |\Pr(X'_k=i)- \Pr(X_k=i)|H(X_{k+1}|X_k=i)\\
    &\leq 2\epsilon \log |\Omega|
\end{align*}
Now the distance $|H(X'_t)-H(X_t)| $ can be controlled by standard bounds which relate statistical difference to maximal entropy difference. To be exact, for $X'_t, X_t$ $\epsilon$-close in statistical difference, 
$$
|H(X'_t)-H(X_t)| \leq H(\epsilon) + \epsilon \log |\Omega|
$$
Combining these bounds, we have that $|H(X_t, \dotsc X_n) - H(X'_t, \dotsc X'_n)| \leq 2n\epsilon \log|\Omega|$ so that 
$$
H(X'_t, \dotsc X_n) - 2n\epsilon \log |\Sigma| \leq H(X_0 \dotsc X_n) \leq H(X'_t, \dotsc X'_n) + t\log|\Omega| + 2n\epsilon \log |\Omega|
$$
Dividing by $n$ and taking the limit as $n\to \infty$ then shows that the entropy rate of a Markov model, regardless of the initial distribution converges to the same value. 
\newline 
The first part of the proposition now follows from the fact that when starting from the stationary probability distribution $\Pi$, we have $H(X_t | X_{<t}) \leq H(X_{t-1}| X_{< t-1})$.
More precisely, assuming $Z_0 \sim \Pi$ we have $Z_1 \sim \Pi$ and so $(X_1,\ldots,X_t)$ is distributed identically to $(X_0,\ldots,X_{t-1})$. We also have $H(X_t | X_0,\ldots,X_{t-1}) \leq H(X_t | X_1,\ldots,X_{t-1}) = H(X_{t-1}| X_1,\ldots,X_{t-1})$. (The inequality comes from ``conditioning does not increase entropy" and the equality comes from the identity of the distributions). 
It follows that the sequence $v_t$, $v_t \triangleq \frac1t \cdot H(X_0,\ldots,X_t)$ is a non-increasing sequence in the interval $[0,\log|\Sigma|]$ and so the limit $\lim_{t \to \infty} \{v_t\}$ exists.
\end{proof}

\begin{definition}[Good Compressor, Universality] For $\epsilon,\delta > 0$, a compression algorithm $A:\Sigma^* \to \{0,1\}^*$ is an {\em $(\epsilon,\delta)$-Good Compressor} for a hidden Markov model $\calm$ of entropy rate $H(\calm)$ if there exists $n_0$ such that for all $n \geq n_0$ we have
$$\Pr_{X_0,\ldots,X_{n-1}} [ |A(X_0,\ldots,X_{n-1})| \leq H(\calm)\cdot (1 +\epsilon)\cdot n ] \geq 1 - \delta.$$
We say that $A$ is 
{\em Universal} (for the class of HMMs) if it is an $(\epsilon,\delta)$-good compressor for every hidden Markov model $\calm$ and every $\epsilon,\delta > 0$. 
\end{definition}

\subsection{The Lempel-Ziv Algorithm}\label{ssec:lz-alg}

\newcommand{\mbin}{\mathrm{bin}}
We describe the algorithm structurally below. First we fix a prefix free encoding of the positive integers denoted $[\cdot]$ such that for every $i$, $|[i]| \leq \log i + 2\log\log i$. We also fix an arbitrary binary encoding of $\Sigma \cup \{\lambda\}$, denoted $\mbin$ satisfying $\mbin(b) \in \{0,1\}^\ell$ where 
$\ell = \lceil \log (|\Sigma|+1) \rceil$.

Given $\bar{X} = X_0,\ldots,X_{n-1}$, write $\bar{X} = \sigma_0 \circ \sigma_1 \circ \cdots \circ \sigma_m$ where $\sigma_i \in \Sigma^*$ satisfy the following properties: 
\begin{enumerate}
    \item[(1)] $\sigma_0 = \lambda$ (the empty string)
    \item[(2)] For $i\in [m-1]$, $\sigma_i$ is the unique string such that
    \begin{enumerate}
        \item[(2a)] $\sigma_i \ne \sigma_{i'}$ for $i' < i$
        \item[(2b)] $\sigma_i = \sigma_{j_i} \circ b_i$ for some $j_i < i$ and $b_i \in \Sigma$.
    \end{enumerate}
    \item[(3)] Finally, for $\sigma_m$, we require $\sigma_m = \sigma_{j_m} \circ b_m$ where $b_m \in \Sigma \cup \{\lambda\}$.
\end{enumerate}

The Lempel-Ziv encoding $\LZ(\bar{X})$ is then $ [i_1]\circ \mbin(b_1) \circ  [i_2] \circ \mbin(b_2) \circ \cdots [i_m] \circ \mbin(b_m)$. 


\subsection{Main Theorem}\label{ssec:main-thm}

\begin{theorem}\label{thm:main} The Lempel-Ziv Algorithm is universal for the class of Hidden Markov Models.
\end{theorem}

\section{Overview of Proof}

The essence of the analysis is quite simple. We first design a simple algorithm, that we call the Iterated Huffman algorithm (see Section~\ref{sec:IH}) to compress from a {\em known} hidden Markov model. This algorithm works using the knowledge of the statistics (the frequencies) of small strings produced by the source. Given the frequencies of all length $L$ sequences, for an appropriately chosen parameter $L$, the Iterated Huffman algorithm builds a Huffman coding scheme for these length $L$ sequences and then given a long sequence of length $n$ that is a multiple of $L$, it divides the string into $n/L$ blocks of length $L$ and applies the Huffman coding scheme for each block separately. 

It is straightforward, using standard concentration bounds for independent random variables or martingales, and elementary facts about convergence of Markov chains, to show that the Iterated Huffman algorithm yields a good compressor. We show this in Section~\ref{sec:IH-anal}. The crux of the analysis is to turn this into a statement about the performance of the Lempel-Ziv algorithm. This is achieved by observing that the Iterated Huffman algorithm is a ``finite-state'' compressor (it only needs to remember a finite number of characters, specifically the last $L$, to compress a string of length $n$) and that Lempel-Ziv is competitive against any finite state compressor. We show this in Section~\ref{sec:LZ-anal}. Together these two steps conclude the proof of Theorem~\ref{thm:main}.

\section{Iterated Huffman (IH) algorithm}\label{sec:IH}

The Iterated Huffman algorithm is essentially a simple one. Given a hidden Markov model $\calm$ and a length parameter $L$ the iterated Huffman algorithm, denoted $\IH_{L,\calm}$ performs the following steps:
\begin{enumerate}
    \item It first computes the expected frequencies of a string $\gamma \in \Sigma^L$, i.e., the quantity $P_L(\gamma) \triangleq \Pr[X_0,\ldots,X_{L-1} = \gamma]$ where $X_0$ is distributed according to the stationary distribution of $\calm$. 
    \item Next the algorithm computes the Huffman coding function $\Huff:\Sigma^L \to \{0,1\}^*$ that minimizes $\sum_{\gamma \in \Sigma^L} P_L(\gamma)\cdot |\Huff(\gamma)|$.
    In particular, we assume that for every $\gamma$ we have $|\Huff(\gamma)| \leq 1 + \log_2 (1/P_L(\gamma))$.\footnote{We note that such an inequality need not be strictly true for every Huffman coding scheme. But weaker (sub-optimal) coding schemes, and in particular Shannon's coding scheme, do achieve this property and this is all we will need from our function $\Huff$.}
    \item Finally to compress a string $X_0,\ldots,X_{n-1} \in \Sigma^n$, it views the string as a string $\bar{Y} = (Y_0,\ldots,Y_{(n/L)-1})$ where the $Y_j$'s are in $\Sigma^L$ and applies $\Huff$ to each $Y_j$. Thus $\IH_{L,\calm}(X_0,\ldots,X_{n-1}) = \Huff(Y_0) \circ \Huff(Y_1) \circ \cdots \circ \Huff(Y_{(n/L)-1}$. (For simplicity we assume $n$ is divisible by $L$.) 
\end{enumerate} 

We stress that the Iterated Huffman is {\em not} universal since the compression algorithm (and the corresponding decompressor) depend on knowledge of the hidden Markov model $\calm$. We also note that we do not care about the computational efficiency of this algorithm. In particular we make no assertions about the complexity of the first step (though it can be bounded as function of $|\Sigma|$ and $L$). This is not relevant to us since all we want to do is show that this algorithm is a good compressor for $\calm$, and that the Lempel-Ziv algorithm is almost as good as this algorithm. We turn to the analysis of the Iterated-Huffman algorithm next.

\section{Analysis of the Iterated-Huffman algorithm}\label{sec:IH-anal}

In this section we show that if we fix the hidden Markov model $\calm$, and the error parameters $\epsilon$ and $\delta$ then there exists an $L$ such that $\IH$ is a good compressor for $\calm$. 

Fix a hidden Markov model $\calm$ and let $X_0,\ldots,X_L,\ldots$ denote a random sequence drawn according to $\calm$ with initial distribution being the stationary distribution for $\calm$. Let $Z_0,\ldots,Z_L,\ldots,$ denote the corresponding state sequence.
For positive integer $L$ and sequence $\gamma \in \Sigma^L$, let $P_L(\gamma)$ denote the probability that $(X_0,\ldots,X_{L-1}) = \gamma$. 
We say that $L$ is {\em $\epsilon$-compressive} if 
$H(X_0,\ldots,X_{L-1}) + 1 \leq H(\calm)(1 + \epsilon)L$. (In particular this implies that the expected length of the Huffman coding of $X_0,\ldots,X_{L-1}$ is upper bounded by $H(\calm)(1 + \epsilon)L$.

\begin{proposition}
For every $\calm$ and $\epsilon$ there exists $L_0$ such that for all $L \geq L_0$, $L$ is $\epsilon$-compressive. 
\end{proposition}
\begin{proof}
By the definition of the entropy rate there exists $L_1$ such that for all $L \geq L_1$ we have 
$$H(X_0,\ldots,X_{L-1}) \leq H(\calm)(1 + \epsilon/2)L.$$ Now let $L_0 = \max\{L_1, 2/(\epsilon\cdot H(\calm)\}$, then for all $L \geq L_0$ we have $$H(X_0,\ldots,X_{L-1}) + 1 \leq H(\calm)(1 + \frac{\epsilon}{2})L + \frac{H(\calm)\epsilon}{2} L = H(\calm)(1+\epsilon)L$$.
\end{proof}
Note that the $\epsilon$-compressive condition also gives a lower bound on $L$ since $H(X_0,\ldots,X_{L-1}) \geq L H(\calm)$ and so we get $1 \leq \epsilon H(\calm) L$. (This fact will be used later.)

If the sequences $Y_0,Y_1,Y_2,\ldots,Y_{n/L}$ were drawn i.i.d. with $Y_i$ having same distribution as $X_0,\ldots,X_{L-1}$ then we would be done immediately by some Chernoff bound arguments --- the expected length of the encoding of each $Y_i$ is $H(\calm)(1 + \epsilon)L$ and $n/L$ independent samples would have total length sharply concentrated around $(n/L) \cdot H(\calm)(1 + \epsilon)L = n \cdot H(\calm)\cdot (1 + \epsilon)$. But the $Y_i$'s are not i.i.d. The rest of the argument below shows that there is enough independence among them to get the same effect.

For states $a,b \in [k]$ let $\rho_{ab,L}$ denote the probability that that $Z_0 = a$ and $Z_L = b$ when the initial distribution of $Z_0$ is the stationary distribution.  Let $P_{ab,L}(\gamma)$ be the probability that $X_0,\ldots,X_{l-1})=\gamma$ conditioned on $Z_0 = a$ and $Z_L = b$. Note that this implies $P_L(\gamma) = \sum_{a,b \in [k]} \rho_{ab,L} \cdot P_{ab,L}(\gamma)$.

Let $\rho_{b|a,L}$ denote the probability that $Z_L = b$ conditioned on $Z_0 = a$. Let $\rho_a = \Pi(a)$ denote the stationary probability of state $a$. Note that we have $\rho_{ab,L} = \Pi(a) \cdot \rho_{b|a,L}$. Further note that as $L \to \infty$, the quantities $\rho_{b|a,L}$ converge to $\Pi(b)$ and the quantities $\rho_{L,ab}$ converge to $\Pi(a)\cdot\Pi(b)$. We say that $L$ is {\em $\epsilon$-mixing} if for every $a$, it is the case that $\sum_{b} | \rho_{ab,L} - \Pi(a)\Pi(b)| \leq \epsilon\cdot \Pi(a)$ or equivalently
$\sum_b | \rho_{b|a,L} - \Pi(b)| \leq \epsilon$.

\begin{proposition}
For every hidden Markov model $\calm$ with an underlying irreducible and aperiodic chain, for every $\epsilon > 0$ there exists $L_0$ such that for all $L \geq L_0$, $L$ is $\epsilon$-mixing. 
\end{proposition}
\begin{proof}
Follows from the fact that irreducible aperiodic Markov chains converge to their stationary distribution. 
\end{proof}

The rest of the proof (simple modulo some standard concentration bounds) shows that if $L$ is $\epsilon_1$-compressive and $\epsilon_2$-mixing for small enough $\epsilon_1$ and $\epsilon_2$ then $\IH$ is $(\epsilon,\delta)$-good.

\begin{lemma}\label{lem:conc}
For every $\calm$ and $\epsilon > 0$, there exist $\epsilon_1>0$ and $\epsilon_2 > 0$ such that if $L$ is $\epsilon_1$-compressive and $\epsilon_2$-mixing, then for every $\delta > 0$  we have $\IH_{L,\calm}$ is $(\epsilon,\delta)$-good for $\calm$.\footnote{So $L$ does not  depend on $\delta$ --- only $n$ does. In fact, the dependence of $n$ on $\delta$ is just logarithmic in $1/\delta$.}
\end{lemma}

\begin{proof}
We assume $\epsilon \leq 1$.

Let $\epsilon_1 = \epsilon/6$. This ensures $H(X_0,\ldots,X_{L-1}) \leq H(\calm)(1+\epsilon/6)L$ and $L \geq 6/(\epsilon H(\calm))$. (Throughout this proof we assume that $Z_0,Z_1,Z_2,\ldots$ denote the underlying states of $\calm$ and that $Z_0$ is distributed according to the stationary distribution $\Pi$.) In what follows we show that for sufficiently small but positive $\epsilon_2$ and for sufficiently large $n$ that is a multiple of $L$, the expected length of $\IH_{L,\calm}(X_0,\ldots,X_{n-1})$ is, with probability at least $1-\delta$, bounded from above by $(1+\epsilon/6)^2(n/L)H(X_0,\ldots,X_{L-1})$. This will establish the lemma.

Recall that we divide time into $n/L$ ``epochs'' starting at multiples of $L$ and of length $L$. For states $a$ and $b$, let $n_a$ denote the number of epochs starting at state $a$, i.e, $n_a = \{0 \leq i < n/L | Z_{iL} = a\}$,
and let $n_{ab}$ denote the number of epochs starting in state $a$ and ending in state $b$, i.e., $n_{ab} = \{0 \leq i < n/L | Z_{iL} = a \mbox{ and } Z_{(i+1)L} = b\}$. Finally for states $a,b$ and $\gamma \in \Sigma^L$, let $K_{ab}(\gamma)$ denote the number of epochs starting in state $a$, ending in state $b$ and generating the output $\gamma$, i.e., 
$$K_{ab}(\gamma) = \{0 \leq i < n/L | Z_{iL} = a, ~ Z_{(i+1)L} = b\mbox{ and } (X_{iL},\ldots,X_{(i+1)L-1}) = \gamma\}.$$ 
Note that in terms of the above quantities we have 
$H(X_0,\ldots,X_{L-1}) =  \sum_{\gamma \in \Sigma^L} P_L(\gamma)\log (1/P_L(\gamma))$, whereas the length of the compression is at most 
$$\sum_{\gamma \in \Sigma^L}  \sum_{a,b \in [k]} K_{ab}(\gamma)\cdot (1 + \log (1/P_L(\gamma))) = (n/L) + \sum_{\gamma \in \Sigma^L}  \sum_{a,b \in [k]} K_{ab}(\gamma)\cdot \log (1/P_L(\gamma)).$$ 
Since we want to show that the latter quantity can be bounded in terms of the former, it suffices to show that $n/L \leq (\epsilon/6)n H(\calm)$ and that with probability at least $1 - \delta$ the following holds: ``for every $a, b, \gamma$, $K_{ab}(\gamma) \leq (1+\epsilon/6) \rho_{ab,L} P_{ab,L}(\gamma)(n/L)$''. The first inequality follows from $L \geq 6/(\epsilon H(\calm))$ which in turn is a consequence of $L$ being $\epsilon_1$-compressive. We thus turn to bounding the $K_{ab}(\gamma)$'s. In what follows we show that the quantities $n_a$, $n_{ab}$, and ultimately $K_{ab}(\gamma)$ are sharply concentrated around their expectation --- all of these concentrations will follow from the $\epsilon_2$-mixing property. 

\begin{claim}\label{clm:conc}
There exists $n_1$ such that for all $n \geq n_1$, the probability that there exists $a,b$ such that either $n_a \not\in \{\frac{n}{L}(\Pi(a)-2\epsilon_2),\frac{n}{L}(\Pi(a)+2\epsilon_2)\}$ or $n_{ab} \not\in \{\frac{n}{L}(\rho_{ab,L}-3\epsilon_2), \frac{n}{L}(\rho_{ab,L}+3\epsilon_2)\}$ is at most $\delta/3$.
\end{claim}

\begin{proof}
Fix a pair $a,b$ and consider the four events 
\begin{enumerate}
    \item[(E1)] $n_a < \frac{n}{L}(\Pi(a)-2\epsilon_2)$
    \item[(E2)] $n_a > \frac{n}{L}(\Pi(a)+2\epsilon_2)$
    \item[(E3)] $n_{ab} < \frac{n}{L}(\rho_{ab,L}-5\epsilon_2)$
    \item[(E4)] $n_{ab} > \frac{n}{L}(\rho_{ab,L}+5\epsilon_2)$
\end{enumerate}

For every one of these events, we prove that they occur with probability at most $\delta/(12k^2)$ and the claim follows by a union bound over the $k^2$ choices of $a$ and $b$ and the four choices of the error events among (E1)-(E4). We thus turn to bounding the probability of these four events for a fixed $a,b$.

We start with event (E1). Let $m = n/L$ denote the number of epochs. Note that the expectation of $n_a$ is $\Pi(a)\cdot m$ and we wish to bound the probability that $n_a$ is smaller than its expectation by an additive $2\epsilon_2 m$.  Let $U_i$ be the indicator of the event that
the $i$th epoch starts in state $a$, i.e., $U_i = 1$ if $Z_{iL} = a$ and $0$ otherwise. Note $n_a = \sum_{i=0}^{m-1} U_i$. We now use the 
the $\epsilon_2$-mixing assumption, to note that $$\Pr[U_i = 1 | U_0,\ldots,U_{i-1} ] \geq \Pi(a) - \epsilon_2.$$ Specifically, if $Z_{(i-1)L} = b$ then $\Pr[U_i = 1 | U_0,\ldots,U_{i-1} U_{(i-1)L}] = \rho_{a|b,L} \geq \Pi(a) - \epsilon_2$ which is a bound that holds for every $b$.

Thus if we create a new sequence of random variables $U'_i$ derived from $U_i$ by setting $U'_i  = 0$ if $U_i = 0$ and $U'_i = \bern((\Pi(a) - \epsilon_2)/\E[U_i])$ then we get that the variables $V_i = \sum_{j<i} U'_j + (m-i)(\Pi(a) - \epsilon_2)$ form a martingale sequence with bounded difference (since $|V_i - V_{i-1}| \leq 1$). Applying Azuma's inequality \cite[Theorem 12.4]{MitzUpf} we get that $$\Pr[V_m  < V_0 - \epsilon_2 m] \leq  \exp(-\epsilon_2^2 m) \leq \frac{\delta}{12k^2}$$ provided $m = n/L$ is sufficiently large (in particular
choosing $n_1 = O(L/\epsilon_2^2 \log (k/\delta)$ suffices). We are now done, since we have $V_0 = m(\Pi(a) - \epsilon_2)$, $V_m = \sum_{i< m} U'_i$ and $U'_i \leq U_i$. Combining with $n_a = \sum_{i<m} U_i$ we get $$\Pr[n_a \leq \frac{n}{L}(\Pi(a)-2\epsilon_2)] \leq \Pr[V_m \leq \frac{n}{L}(\Pi(a)-2\epsilon_2)] \leq \frac{\delta}{12k^2}.$$

The bound for (E2) is completely similar. The analyses of (E3) and (E4) are also similar with minor differences. We now define the random variable sequence $W_i$ where $W_i = 1$ if $Z_{iL}=a$ and $Z_{(i+1)L} = b$, and $W_i = 0$ otherwise.
Note that we have $\sum_{i=0}^{m-1} W_i = n_{ab}$ and so one may hope for an analysis as in
the case of (E1), however $W_i$ is not sufficiently independent of $W_{i-1}$ to reproduce the same steps. Instead we bound the even terms $\sum_{0 \leq i< m/2} W_{2i}$ and odd terms
$\sum_{0 \leq i < m/2} W_{2i+1}$ separately. In each case we now have $W_i | W_{i-2}$ has enough independence to claim that $$\E[W_i | W_{i-2}] \geq (\Pi(a) - \epsilon_2)(\Pi(b) - \epsilon_2) \geq \Pi(a)\Pi(b) - 2\epsilon_2.$$ This allows us to conclude that if $n$ is sufficiently large then $$\Pr[\sum_i W_i \leq m(\Pi(a)\Pi(b) - 3\epsilon_2)] \leq \frac{\delta}{12k^2}.$$ This concludes the proof of Claim~\ref{clm:conc}.
\end{proof}

We now turn to showing that for every $\gamma \in \Sigma^L$ the empirical count of $\gamma$ given by $\sum_{a,b}K_{ab}(\gamma)$ concentrates around its expectation given by $\frac{n}{L}\BP(\gamma) = \frac{n}{L}\sum_{a,b}\rho_{ab,L}P_{ab,L}(\gamma)$. 
Note that conditioned on $n_{ab}$, $K_{ab}(\gamma)$ can be expressed as a sum of $n_{ab}$ i.i.d. random variables distributed according to $\bern(P_{ab,L}(\gamma))$. Thus, by concentration $$\Pr[K_{ab}(\gamma) \geq (1+\epsilon_2) n_{ab}P_{ab,L}(\gamma)] \leq \exp(-\epsilon_2^2 n_{ab}P_{ab,L}(\gamma)).$$ By setting $n$ large enough, we get that this quantity is at most
$\delta/(2k^2|\Sigma|^L)$ provided $n_{ab} \geq (\rho_{ab,L}-3\eps_2)(n/L)$. (Specifically we will need $n = \max_{a,b,\gamma} \{\Omega(L/(\epsilon_2^{2} P_{ab,L}(\gamma) \rho_{ab,L}) \log(3k^2 |\Sigma|^L/\delta)\}$.\footnote{Note that the bottleneck here is likely to be the $1/P_{ab}(\gamma)$ term which is at least exponential in $L$ with a base that may depend on $\calm$.}
By a union bound over $a,b$ and $\gamma$, we conclude that the probability that there exists
$a,b,\gamma$ such that $K_{ab}(\gamma) \geq (1+\epsilon_2) n_{ab}P_{ab,L}(\gamma)$ is at most $\delta/2$. Combining with Claim~\ref{clm:conc} we get that with probability at least $1-\delta$ we have  
$n_{ab} \in \{\frac{n}{L}(\rho_{ab,L}-3\epsilon_2), \frac{n}{L}(\rho_{ab,L}+3\epsilon_2)\}$ and $K_{ab}(\gamma) \leq (1+\epsilon_2) n_{ab}P_{ab,L}(\gamma)$. When these hold we now claim that the length of the compression is at most $(1+\epsilon)H(\calm)n$. We first note that the conditions ensure 
$$K_{ab}(\gamma) \leq (1+\epsilon_2)P_{ab,L}(\gamma)\frac{n}{L}(\rho_{ab,L}+3\epsilon_2)
\leq (\rho_{ab,L}+4\epsilon_2)P_{ab,L}(\gamma)\frac{n}{L}.$$
We now set $\epsilon_2 = \min_{a,b} \{ \epsilon/(24 \rho_{ab,L})\}$ so that we have
$K_{ab}(\gamma) \leq (1+\frac{\epsilon}{6})\rho_{ab,L}P_{ab,L}(\gamma)\frac{n}{L}$.
Summing over $a,b$ we get that the total number of occurrences of $\gamma$ is 
$$\sum_{a,b} K_{ab}(\gamma) \leq (1+\frac{\epsilon}{6}) \frac{n}L \sum_{a,b} \rho_{ab,L}P_{ab,L}(\gamma) = (1+\frac{\epsilon}{6}) \frac{n}L P_L(\gamma).$$
We conclude that the length of the compression 
$$|\IH_{L,\calm}(X_0,\ldots,X_{n-1})| \leq \sum_{\gamma} (1+\frac{\epsilon}{6}) \frac{n}L P_L(\gamma) (1+ \log (\frac{1}{P_L(\gamma)}) = (1+\frac{\epsilon}{6})(n/L + H(X_0,\ldots,X_{L-1})).$$
Finally we conclude by using $n/L \leq \frac{\epsilon}{6} H(\calm) n$ and $H(X_0,\ldots,X_{L-1}) \leq (1+\frac{\epsilon}{6})H(\calm)L$ (both of which follow from the fact that $L$ is compressive).
Putting the inequalities together we have $|\IH_{L,\calm}(X_0,\ldots,X_{n-1})| \leq (1+\frac{\epsilon}{6})^3 H(\calm) n \leq (1+\epsilon) H(\calm) n$ (where the last inequality uses $\epsilon \leq 1$) and thus we have that $\IH$ is an $(\epsilon,\delta)$-good compressor for $\calm$. 
\end{proof}

\section{Analysis of Lempel Ziv}\label{sec:LZ-anal}

\begin{definition}[Finite state transducer]
A \textit{finite state transducer} is a finite state machine with a single input tape and a single output tape. That is, at every step, the transducer shifts states based on the input and then writes some symbols to the output tape. Formally a transducer is given by a 5-tuple $(Q,q_0,\Sigma,\Gamma,\delta)$ where $Q$ is a finite set representing the state space, $q_0 \in Q$ is the initial state, $\Sigma$ is a finite set representing the input alphabet and $\Gamma$ is a finite set representing the output alphabet and $\delta:Q \times \Sigma \to Q \times \Gamma^*$ represents the actions of the transducer. Specifically on input $(X_1,\ldots,X_n)$ with $X_i \in \Sigma$ the output of the transducer is $Y_1 \circ \cdots \circ Y_n$ where $Y_i \in \Gamma^*$ are derived by setting $(q_i,Y_i) = \delta(q_{i-1},X_i)$ inductively for $i \in [n]$. 
\end{definition}

\begin{definition}[Finite State Compressors]
Algorithm A is a finite state compressor if there exists an integer $S$ such that for every
$n$ there exists a $S$-state transducer $T_n$ such that for every string $X\in \Sigma^n$ we have $A(X) = T_n(X)$. 
\end{definition}

Note that ``finite state compressors'' are not necessarily constructive since the transducer is allowed to depend arbitrarily on the length of the string being compressed. Nevertheless this notion turns out to be very useful. We 
\begin{proposition}\label{prop:ih-fsc}
For every integer $L$, $\IH_{L,\calm}$ Algorithm is a finite state compressor.
\end{proposition}
\begin{proof}
Given the length $n$ of the string to be compressed, recall that $\IH_{L,\calm}$ uses the expected frequency vector $\{P_L(\gamma)\}_{\gamma \in \Sigma^L}$ to produce a Huffman coder  $\Huff:\Sigma^L \to \{0,1\}^*$ corresponding to this frequency vector. The output of $\IH_{L,\calm}(X)$ is the just the repeated (iterated) application of the Huffman code to the blocks of are the partition of $X$ into length $L$ sequences. 

The finite state compressor captures this function $\Huff$ using $O(|\Sigma^L|)$-states. These states correspond to nodes of a $|\Sigma|$-ary tree of depth $L$, which record the symbols seen in the current block. At the leaf corresponding to the block $\gamma$, the finite machine outputs $\Huff(\gamma)$ and returns to the root of the tree (to process the next block). 
\end{proof}

\begin{lemma}\label{lem:fsc-comp}
Lempel-Ziv is competitive against finite state compressors. Specifically if $A$ is an finite state compressor then for every $X \in \Sigma^n$, $|\LZ(X)| \leq |A(x)|\cdot (1 + o_n(1))$. 
\end{lemma}

Before proving the lemma above we introduce a final, and key, quantity that identifies the ``complexity'' of a string. This complexity will give an upper bound on the Lempel-Ziv coding length and a lower bound on the compression length under any finite state compressor.

\begin{definition}\label{def:comp}
For a string $X \in \Sigma^n$ we define its complexity $C(X)$ to be the largest integer $t$ such that there exist $t$ distinct strings $Y_1,\ldots,Y_t \in \Sigma^*$ such that $X = Y_1 \circ Y_2 \circ \cdots \circ Y_t$, i.e., $X$ can be written as concatenation of $t$ distinct strings.
\end{definition}

It is easy to give a (seemingly crude) upper bound on the length of the Lempel-Ziv coding of a string in terms of its complexity.

\begin{proposition}\label{prop:fsc-up}
For every string $X \in \Sigma^n$ with complexity $C(X) = t$, its length under the Lempel-Ziv coding is at most $t \log t + O(t \log \log t) = t \log t \cdot (1 +o_n(1))$.
\end{proposition}


\begin{proof}
 Note that the Lempel-Ziv algorithm produces a decomposition of the form $\sigma_1 \circ \cdots \circ \sigma_m$ where the $\sigma_j$'s are distinct. So $m \leq t$. The encoding length is at most 
$m \cdot (\log m + 2\log \log m + \log_2 | \Sigma |)$ where the $\log m + 2\log \log m$ is for the prefix free coding of integers from $0$ to $m$ and the the $\log |\Sigma|$ bits are needed to describe an element of $\Sigma$. Thus this length is at most $t \log t + O(t \log \log t)$. 

This expression is already of the form $t\log t\cdot (1 + o_t(1))$. To conclude we only need to show that $t \to \infty$ as $n \to \infty$. In fact it is easy to see that $t \geq \sqrt{n}$. This follows from the fact that every string $X$ can be decomposed into $\sigma_1 \circ \sigma_2 \circ \cdots \sigma_{\sqrt{n}}$ where $|\sigma_i| = i$ for $i < \sqrt{n}$ and $|\sigma_{\sqrt{n}}| \geq \sqrt{n}$. Clearly the $\sigma_i$'s are distinct (since their lengths are all distinct) and so we get a decomposition of $X$ into $\sqrt{n}$ distinct strings establishing $C(X) \geq \sqrt{n}$. 

We thus have that the length of the Lempel-Ziv encoding is $t \log t + O(t\log \log t) = t \log t (1 + o_n(1))$, and the proposition follows.
\end{proof}

The next lemma gives a lower bound on the compression length for finite state compressors in terms of the complexity of the string being compressed.

\begin{lemma}\label{lem:fsc-low}
For every string $X \in \Sigma^n$ with complexity $C(X) = t$ and for every finite state compressor $A$ with $s$ states, the length $|A(X)|$ is at least $t \log t - (3 + 2\log s)\cdot t = t\log t\cdot (1 - o_n(1))$.
\end{lemma}

\begin{proof}
Let $S$ denote the states of the finite state compressor $E$. We use $E$ to denote the finite state machine doing the compression. Without loss of generality we assume the outputs of $E$ occur on the transitions (and not the states). 

Let $X = Y_1 \circ Y_2 \circ \cdots \circ Y_t$ with $Y_i$'s being distinct. Note that as $E$ compresses $X$, it encounters the string $Y_i$ at some state and compresses this part. Let $a(i)$ denote the state it starts in when parsing $Y_i$ and let $b(i)$ denote the state it ends at after parsing $Y_i$. Let $Z_i \in \{0,1\}^*$ be the output of $E$ during this phase. 

We partition $[t]$ into sets $\{\Pi_{ab}\}_{a,b \in S}$ as follows: we let
$\Pi_{ab} = \{i | a(i) = a \mbox{ and } b(i) = b\}$. The key to our analysis is the following claim.

\begin{claim}
Let $i,j \in \Pi_{a,b}$ for some $a,b$. Then $Z_i \ne Z_j$. 
\end{claim}

\begin{proof}
This is easy to see. Suppose $Z_i = Z_j$. Now consider the string $X'$ which looks like $X$ except the positions of $Y_i$ and $Y_j$ are flipped, i.e., $X' = Y'_1 \cdots Y'_t$ where $Y'_\ell = Y_\ell$ for $\ell \not\in \{i,j\}$ and $Y'_i = Y_j$ and $Y'_j = Y_i$. 
Then the compression $T_n(X') = T_n(X)$ which rules out correct decompression. 
\end{proof}

Now we are essentially done, modulo some calculations. First let $t_{ab} = |\Pi_{a,b}|$. We first argue that $\sum_{i \in \Pi_{a,b}} |Z_i| \geq t_{ab} \cdot (\log_2 t_{a,b} - 3)$. To see this we use the fact that the $Z_i$'s must be distinct for $i \in \Pi_{a,b}$ and there are at most $2^i$ distinct binary strings of length $i$. Letting $t_{ab} = 2^k + m$ where $0 \leq m < 2^k$ we have $$\sum_{i \in \Pi_{a,b}} |Z_i| \geq \left(\sum_{j=0}^{k-1} j \cdot 2^j\right) + k \cdot (m+1).$$ Using the simplification $\sum_{j=0}^{k-1} j \cdot 2^j = (k-2) \cdot 2^{k} + 1$, we now get $$\sum_{i \in \Pi_{a,b}} |Z_i| \geq (k-2) 2^k + k (m+1) \geq (k-2) t_{ab} \geq (\log t_{ab} - 3) \cdot t_{ab}.$$

Finally to get a bound in terms of $t$ we use the fact that $t = \sum_{a,b \in S} t_{ab}$. Convexity of $x\log x$ function and Jensen's inequality now imply that
$\sum_{a,b} (\log t_{ab} - 3) \cdot t_{ab}$ is lower bounded by $$\log (\sum_{a,b} \frac{t_{a,b}}{|S|^2}) - 3 \sum_{a,b} t_{a,b} = t (\log (\frac{t}{|S|^2}) - 3).$$ 
\end{proof}

\begin{proof}[Proof of Lemma~\ref{lem:fsc-comp}]
By Proposition~\ref{lem:fsc-comp} we have that $|\LZ(X)| \leq t\log t \cdot (1 + o_n(1))$ where $t = C(X)$ is the complexity, as in Definition~\ref{def:comp} of the string $X$. 
By Lemma~\ref{lem:fsc-low}, we have $|A(x)| \geq t\log t \cdot (1 - o_n(1))$ or equivalently $t\log t = |A(X)| \cdot (1 + o_n(1))$. We conclude that $|\LZ(X)| \leq |A(X)|\cdot (1 + o_n(1))$. 
\end{proof}

\begin{proof}[Proof of Theorem~\ref{thm:main}]
Fix a Markov model $\calm$. By Lemma~\ref{lem:conc} we have that the Iterated Huffman $\IH_{L,\calm}$ algorithm is an $(\epsilon/3,\delta)$-good compressor for sufficient large $L$, i.e., for all large enough $n$ we have that $\Pr_{X}[|\IH_{L,\calm}(X)| \geq H(\calm)(1+\epsilon/3)n] \leq \delta$ for $X \in \Sigma^n$ drawn from $\calm$. 

By Proposition~\ref{prop:ih-fsc} we have that $\IH_{L,\calm}$ is a finite state compressor. And thus by Lemma~\ref{lem:fsc-comp}, we have for every $X$, $|\LZ(X)| \leq |\IH_{L,\calm}(X)| \cdot (1 + o_n(1))$. In particular for large enough $n$ we have $|\LZ(X)| \leq |IH_{L,\calm}(X)| \cdot (1+\epsilon/3)$.

Combining the two inequalities above we have that with probability at least $1 - \delta$, we have 
$$|\LZ(X)| \leq |\IH_{L,\calm}(X)| \cdot (1+\epsilon/3) \leq H(\calm)(1 + \epsilon/3)^2 n \leq H(\calm) (1+ \epsilon)n,$$ (where the last inequality uses $\epsilon \leq 1$) 
thus yielding the theorem.
\end{proof}


\end{document}